\documentclass[11pt,centertags,reqno]{amsart}

\usepackage[foot]{amsaddr}
\usepackage[english]{babel}
\usepackage[T1]{fontenc}
\usepackage[numbers]{natbib}
\usepackage{amssymb}
\usepackage{fancyhdr}
\usepackage{url}
\usepackage{hyperref}
\usepackage{verbatim}
\usepackage{leftidx}
\usepackage{color,graphicx}

\usepackage{mathrsfs, mathtools}
\usepackage{stmaryrd}

\usepackage{marvosym}
\mathtoolsset{showonlyrefs}

\newtheorem{Lemma}{Lemma}
\newtheorem{proposition}{Proposition}
\def \P {\mathbf P}
\def \bF {\mathbb F}
\def \F {\mathcal F}
\def \R {\mathbb R}

\def \E {\mathbf E}
\newcommand{\esp}[2][\mathbb E] {#1\left[#2\right]}

\usepackage{color}

\begin{document}

\title[Implicit incentives and Partial information]{Implicit Incentives for Fund Managers with Partial Information}


\author{Flavio Angelini  \and
Katia Colaneri \and
Stefano Herzel         \and
Marco Nicolosi
}\address{Flavio Angelini and Marco Nicolosi, Department of Economics, University of Perugia, Via A. Pascoli 20, 06123 Perugia, IT.}
\address{Katia Colaneri and Stefano Herzel, Department of Economics and Finance, University of Rome Tor Vergata, Via Columbia 2, 00133 Roma, IT.}\email{flavio.angelini@unipg.it}\email{katia.colaneri@uniroma2.it}\email{stefano.herzel@uniroma2.it}, \email{marco.nicolosi@unipg.it}



\date{}

\begin{abstract}
We study the optimal asset allocation problem for a fund manager whose compensation depends on the performance of her portfolio with respect to a benchmark. 
The objective of the manager is to maximise the expected utility of her  final wealth.
The manager observes the prices but not the values of the market price of risk that drives the expected returns.  The estimates of the market price of risk get more precise as more observations are available.
We formulate the problem as an optimization
 under partial information.
The particular structure of the incentives makes the objective function  not concave. We solve the problem via the martingale method and, with a concavification procedure, we obtain the optimal wealth and the  investment strategy.
A numerical example shows the effect of learning on the optimal  strategy.


\end{abstract}

\maketitle

\section{Introduction}\label{s1}
The reward of a fund manager usually increases when the Asset Under Management (AUM) grows, while it decreases when the AUM shrinks.  The AUM may grow either because of a higher  value of the assets  or because  of new money flowing into the fund. Good  performances of the fund with respect to its relative benchmark are likely to attract new investors. Therefore, contracts based on the AUM create an implicit incentive for the manager to beat the benchmark.
We study the problem of a portfolio manager whose compensation depends on the AUM modelled through the relative performances with respect to a benchmark. This framework generalizes the setting of Basak, et al. (2007)\cite{BPS2007}, by considering a market model with one risk-free  and one risky asset whose expected returns depend on an unobservable stochastic process, the “market price of risk”. We introduce the realistic assumption  that the manager has a limited knowledge on the market,  she can only observe stock prices  and estimates the market price of risk from them. Therefore, the manager is facing an optimization problem under partial information.

Optimization problem under partial information  are usually solved in two steps: the first step, called reduction, consists of deriving the conditional distribution of the  market price of risk with respect to the observed information flow; the second step solves the equivalent  problem under the observed information. An important feature of our setting is that, while the market of claims contingent to the knowledge of the market price of risk is incomplete, the market restricted to those claims contingent only to stock prices is instead complete. We will exploit this fact to solve the optimization problem applying a martingale approach with the unique equivalent martingale measure (under the restricted setting) and then using a concavification argument to determine the unique optimal solution.

Although many papers have been written  about  the issues of relative incentives and of optimization under partial information (see the literature review provided in Section \ref{lit_rev}), this one is, to the best of our knowledge, the first one to
analyse the combined effect of such issues on the optimal strategy of a portfolio manager.
We contribute to the literature by providing the solution to the optimization problem in semi-closed form and we present one
example where we show that the optimal strategy  depends on the risk aversion of the manager and on the economic situation of the market. When the risk aversion of the manager is larger (lower) than that of a manager equipped with a logarithmic utility, she will tend to decrease (increase)  her investment in the risky asset to hedge against the future adjustment in the estimates of the unknown parameter.

The paper is organized as follows.
After a literature review (Section \ref{lit_rev}), Section 2  presents the market model and the portfolio optimization problem faced by the manager. In Section 3 we solve the optimization problem in two steps. First, we derive the dynamics of the filtered estimate of market price of risk, in order to reduce the problem to a common information flow. This procedure allows to obtain market dynamics driven by a unique source of randomness and hence the market model under partial information turns out to be complete. Second, we apply the martingale method along with concavification to characterize the optimal final wealth and the optimal investments strategy. Section 4 contains a numerical illustration of our results. Conclusions and comments  are provided in Section 5. We relegated  proofs and calculations  to Appendix \ref{optfinval}--\ref{app:ODEs}.

\subsection{Literature review}\label{lit_rev}
 The structure of portfolio managers' compensation is studied for instance in Ma et al. (2019)\cite{MTG2019}, who show that performance based incentives represent the main form of compensation for portfolio managers in the US mutual fund industry. Of course, this is not the only type of incentive for fund managers. Option-like incentives of different nature (as for example management fees, investor's redemption options or funding options by prime brokers) apply in fund managers compensation contract, and influence manager's leverage decisions (see, e.g. Lan et al. (2013)\cite{LWY2013} and Buraschi et al. (2014)\cite{buraschi2014})

Basak et al. (2007)\cite{BPS2007} compute the optimal strategy followed by  the manager under the assumption that she knows exactly the parameters driving the asset price process. They show that, when at an intermediate date the  return of the fund is either very low or very large compared to the benchmark, the manager forgets abut the implicit incentives determined by the fund-flows and reverts to the normal strategy, that is the one determined by Merton (1971)\cite{merton1971}. However, when the
current  return is closer to the benchmark, the manager tilts her strategy from the Merton level to try to beat the  benchmark. 
Nicolosi et al. (2018)\cite{NHA} extend their framework to consider mean-reversion either in the market price of risk or in the volatility.
Basak et al. (2008)\cite{BPS2008} introduce additional restrictions on the set of admissible strategies to contrast the tendency of managers to increase riskiness when their portfolio under-performs the benchmark, in order to align managers' scope to that of investors.
The optimal allocation problem for institutional investors concerned about their performance with respect to a benchmark index is studied in Basak and Pavlova (2013)\cite{BP2013}. The objective there is to show how incentives influence the prices of the assets hold by institutional investors. In particular the authors found that, differently from standard investors, institutions tend to form portfolios of stocks that compose the benchmark index, they push up prices of stocks in the benchmark index by generating excess demand for index stocks and induce excess correlation among these stocks.
Carpenter (2000)\cite{carpenter2000} analyses the optimal investment problem of a risk adverse manager who is compensated with a call option on the asset under management. In this paper the  market model is assumed to be complete and the non-concavity of the objective function is addressed by introducing a {\em concavification} argument and showing that the optimal solution takes values on a set where the original non-concave objective function is equal to the minimal concave
function dominating it. An explicit solution to this problem in the Black-Scholes setting is provided in Nicolosi (2018)\cite{nic2018} while  Herzel and Nicolosi (2019)\cite{HN} extend the solution to the case of mean-reverting processes. The impact of commonly observed incentive contracts on managers’ decisions is also studied in Chen and Pennacchi (2009)\cite{CP2009}, where the authors found that for particular compensation structure, when a fund is performing poorly, the deviation from the benchmark portfolio is larger than in case of good performance.

Other important contributions on the literature of delegated portfolio management problem include  Cuoco and Kaniel (2011)\cite{CK2011}, who investigate the case where managers receive a direct compensation, related to the performance, from investors and discuss asset price implications in equilibrium.
Different compensation schemes have been considered, for instance, in Barucci and Marazzina (2016)\cite{BM2016} in a portfolio optimization problem for a manager who is remunerated through a High Water Mark incentive fee and a management fee and in Barucci et al. (2019)\cite{BMM2019} where a penalty on the remuneration is applied if the fund value falls below a fixed threshold, namely a minimum guarantee.

Optimal asset allocation under partial information has been widely studied in the literature. Brendle (2006) \cite{brendle2006} considered the optimal investment problem for a partially informed investor endowed with bounded CRRA preferences in a market model driven by an unobservable market price of risk via the HJB approach. Hata et al. (2018) \cite{hata2018} also included consumption. A more general setting, not necessarily Markovian, has been analysed for instance in Bj\"ork et al. (2010) \cite{BDL} and Lindensjo (2016)\cite{lind2016}, under the assumption of market completeness. The optimization problem in these papers is solved using the Martingale approach.

The partial information case in a delegated portfolio management has been considered in the recent literature by Barucci and Marazzina (2015)\cite{BarucciMarazzina2015} in a slightly different setting compared to ours, where market is subject to two regimes, modelled via a continuous time two-state Markov chain and in Huang et al. (2012)\cite{HKH2012} where investment learning is studied under a Bayesian approach.

Other contributions in the case where prices are modelled as diffusions are Lackner (1995, 1998)\cite{lackner1995,lackner1998}. Brennan (1998)\cite{brennan1998} and Xia (2001)\cite{xia2001}  study the effect of learning on the portfolio choices, and Colaneri et al. (2020) \cite{CHN2020} address the problem of computing the price that a partially informed investor would pay to access to a better information flow on the market price of risk.
Investment problems in a market with cointegrated assets under partial information are studied in some recent works as for instance Lee and Papanicolaou (2016) \cite{lee2016} and Altay et al. (2018, 2019) \cite{altay2018,altay2019}.


\section {Market model and the portfolio optimization problem}\label{sec:market}
We fix a probability space $(\Omega, \mathcal F, \mathbb{P})$. Let $\mathbb{F}=\{\mathcal{F}_t, t\geq 0\}$ be a complete and right continuous filtration representing the global information. We consider a market model with one risky asset with price $S_t$, the {\em stock}, and one risk-free asset with price $B_t$.
We assume that the price of the risk-free asset follows
\begin{equation}
\frac{d B_t}{B_t} = r   dt \label{generalmodel2}
\end{equation}
with the constant $r>0$ representing the constant interest rate. The risky asset price is modelled by a geometric diffusion
\begin{equation}
\frac{d S_t}{S_t} = \mu_t   dt + \sigma dZ^S_t \label{generalmodel1}
\end{equation}
where $Z^S_t$ is a one dimensional standard Brownian motion, $\sigma>0$ is the constant volatility and the drift is the process
 \begin{equation}
\mu_t = r + \sigma X_t,
\label{drift}
\end{equation}
which depends linearly on the market price of risk $X_t$.
The process $X_t$ satisfies
\begin{equation}
d{X_t} = -\lambda ({X_t}-\bar{X})dt+ \sigma_X dZ_t^X,
\label{XStoc}
\end{equation}
where  $\lambda >0$ is a constant representing  the strength of attraction toward the long term expected mean $\bar{X}$, $\sigma_X > 0$ is  the volatility of
the market price of risk and  $Z^X$ is a one-dimensional standard Brownian motion correlated with $Z^S$ with correlation $\rho \in [-1,1]$.

We assume that the market price of risk is a latent variable that is not directly observed, and its value can only be derived through the observation of $S_t$. That means that the available information  is given by the filtration $\mathbb{F}^S:=\{\F^S_t,\ t \in [0,T]\}$, generated by the process $S$ \footnote{At any time $t$, $\F^S_t$ is the right continuous and complete $\sigma$-algebra generated by the process $S$ up to time $t$. Specifically, $\F^S_t:=\sigma\{S_u, \ 0\leq u \leq t\}\vee \mathcal O$ where $\mathcal O$ is the collection of all $\P$-null sets. Notice that $\F^S_t \subset \F_t$, which models the fact the manager has a restricted information on the market.}.
Let us note that, since there are two risk factors $Z^S$ and $Z^X$, but only one traded asset besides the money market account, this market model is incomplete.

We study the problem of
a fund manager who trades  the two assets, $S_t$ and $B_t$, continuously in time on $[0,T]$, starting from an initial capital $w$. We assume that the stock does not pay dividends before time $T$. We describe the trading strategy of the manager by a process $\theta=\{\theta_t, \ t \in [0,T]\}$ representing the fraction of wealth invested in the risky asset at any time $t \in [0,T]$. We only consider trading strategies that are self-financing and based on the available information, hence defining an {\em admissible} strategy as a self-financing trading strategy, adapted to the filtration $\bF^S$ and, to prevent arbitrage from doubling strategies,  such that
\begin{equation}\label{eq:integrability}
\E\left[\int_0^T\left(|\theta_t X_t|+\theta^2_t\right)   d t\right] <\infty.
\end{equation}
The set of all admissible strategies is denoted by $\mathcal A^{S}$.
The wealth process generated by an admissible strategy $\theta_t$ is
\begin{equation}
\frac{dW_t}{W_t} = (r + \theta_t \sigma X_t)dt+\theta_t\sigma dZ^S_t, \quad W_0=w>0.
\label{BC}
\end{equation}

The manager's compensation is implicitly determined by the value of the AUM at time $T$, according to $f_T(W_T,Y_T) W_T$, where $f_T$ is the new funds flow from investors at time $T$ depending on the relative performance of the portfolio with respect to a benchmark $Y$.
The benchmark $Y$ is
the value of the constant strategy
 $\beta$ and hence it follows
\begin{equation*}
\frac{dY_t}{Y_t} = (r + \beta \sigma X_t)dt+\beta
\sigma dZ^S_t .
\end{equation*}
The continuously compounded returns on the manager's
portfolio and on the benchmark over the period $[0, t]$ are given by
$R^W_t = \ln \frac{W_t}{W_0}$ and
$R^Y_t = \ln \frac{Y_t}{Y_0}$, respectively. To compare relative performances, we  set $Y_0 = W_0$. The difference $R^W_T-R^Y_T$ provides the {\em tracking error} of the final wealth relative to the benchmark.
The funds flow to relative performance relationship is described by function $f_T$
\begin{equation}\label{eq:compensation}
f_T(W_T,Y_T) = \left\{
\begin{array}{ccc}
f_L  &  \text{if}  & R^W_T-R^Y_T<\eta_L  \\
f_L + \psi \cdot(R_T^W -R_T^Y - \eta_L)  & \text{if}   &  \eta_L \le R_T^W
-R_T^Y <\eta_H \\
f_H :=  f_L + \psi \cdot (\eta_H - \eta_L)  &  \text{if}  &   R_T^W -R_T^Y
\ge\eta_H
\end{array}
\right.
\end{equation}
with $f_L > 0, \psi >0$, and $\eta_L \leq
\eta_H$  and it is illustrated in Figure \ref{PayoffFig}.
\begin{figure}[htbp]
\begin{center}
\centerline{\includegraphics[{scale=0.6}]{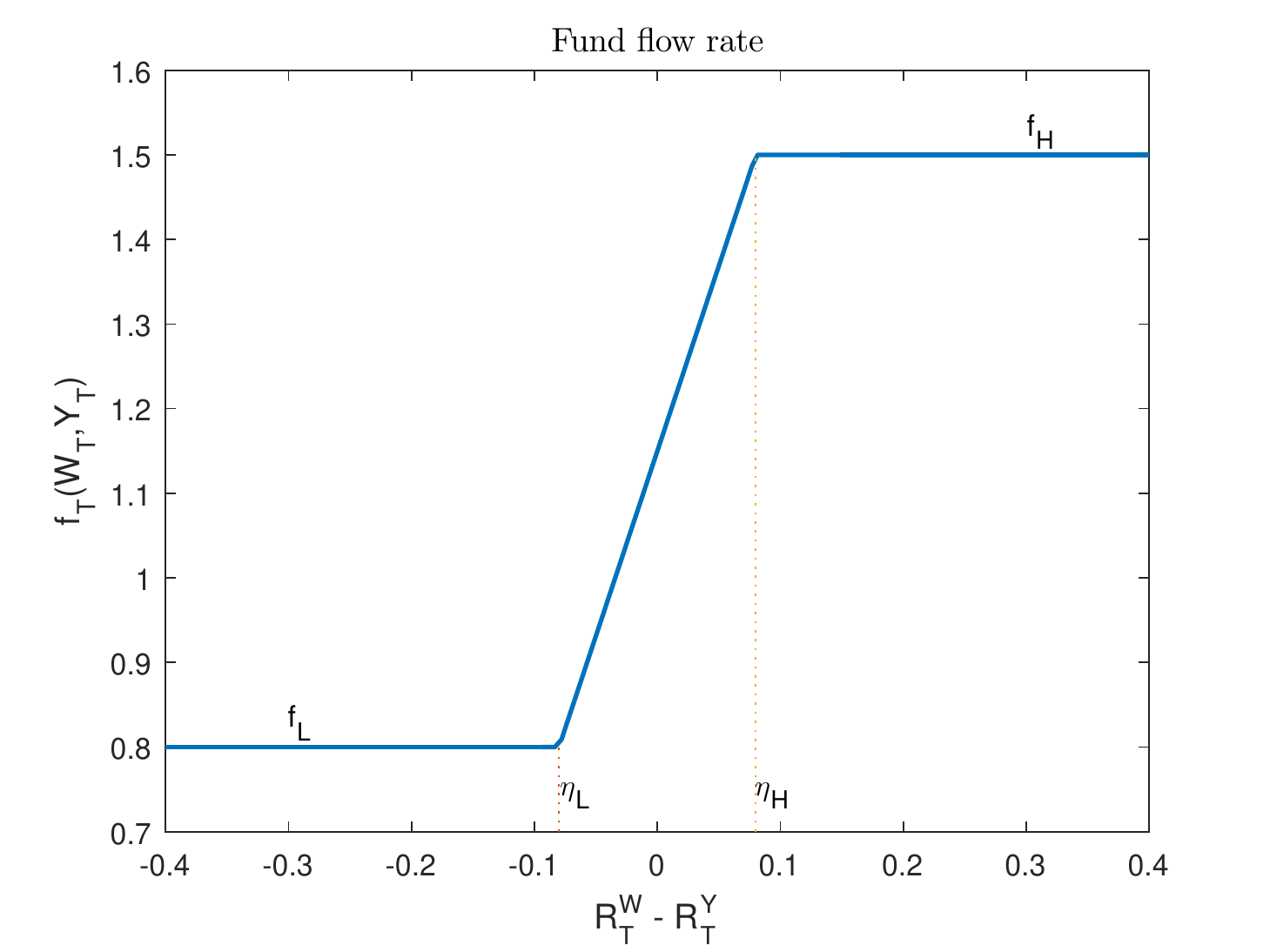}}
\caption{The funds flow $f_T(W_T,Y_T)$ as a function of relative performance $R_T^W -R_T^Y$, with parameters $f_L=0.8$, $f_H=1.5$, $\eta_L=-0.08$ and $\eta_H=0.08$.}
\label{PayoffFig}
\end{center}
\end{figure}
This simplified structure of the funds flow to relative performance relationship, called in the literature {\em collar type}, shows that if the manager return is below the benchmark return of at least $\eta_L$ or above the benchmark return of at least $\eta_H$, the flow rate received by the fund is flat (with different rates $f_L<f_H$). When the relative performance, measured in terms of tracking error, is between $\eta_L$ and $\eta_H$, the flow function is a linear segment with a positive slope. The function $f_T$ also has two kinks when the difference $R_T^W -R_T^Y$ reaches the levels $\eta_L$ and $\eta_H$.
The funds flow to relative performance relationship in Equation \eqref{eq:compensation} was proposed by Basak et al. \cite{BPS2007} to describe an implicit incentive scheme and it is based on the empirical analysis of  Chevalier and Ellison (1997) \cite{chevalier1997}.
The idea is that, if the fund under-performs with respect to the benchmark, investors tend to withdraw their money, the AUM decreases, and the manager receives a lower compensation. The opposite happens in the case of over-performance. Citing Basak et al. \cite{BPS2007}: ``(...) this simple way of modeling fund flows is able to capture most of the insights pertaining risk-taking incentives of a risk averse manager".


The manager maximizes the expected utility of her implicit incentives over the set of  admissible strategies $\mathcal A^{S}$,
\begin{equation}\label{opt_pr}
\max_{\theta \in \mathcal{A}^{S}}E \left[u(W_T f_T (W_T,Y_T))\right]
\end{equation}
with initial budget $W_0=w$.
We assume that  the manager is endowed with a power utility function
\begin{equation}
u(x) = \frac{1}{1-\gamma}x^{1-\gamma},  \label{utility}
\end{equation}
with nonnegative risk aversion parameter $\gamma \neq 1$. The case $\gamma=1$ corresponds to the logarithmic utility.
Since the market price of risk is not observable, this is  an optimization problem under restricted information. To solve it, we first  reduce it to a setting with a common information flow by replacing the unobservable process $X_t$  with its conditional expectation. This standard procedure allows us to consider  an equivalent optimization problem under the available information, see, e.g. Fleming and Pardoux (1982)\cite{fleming1982optimal}. We characterize the conditional expectation of $X_t$ in the next section via Kalman filtering.

\section{Optimal wealth and strategies}\label{sec:full_info}


In this section we solve the  problem \eqref{opt_pr}. The first step is to estimate  the unobservable market price from  stock prices.
Applying the Kalman filtering theory\footnote{Notice that the stock price process $S$ and its log-return generate the same type of information. This is a key feature since the drift of the log-return is a linear function of $X$, and hence the Kalman filter applies. The same setting has been considered for instance in Colaneri et al. (2020)\cite{CHN2020}.} we get that the conditional distribution of  market price of risk is Gaussian with conditional mean $\pi_t=\E\left[X_t | \F^S_t\right]$, and conditional variance  $R_t:= \E\left[\left(X_t-\esp{X_t|\F^S_t}\right)^2|\F^S_t\right]$. To derive $\pi_t$ and $R_t$ we
introduce the {\em innovation} process
\begin{equation}\label{eq:innovation}
I_t:= Z^S_t+\int_0^t (X_u-\pi_u) d u.
\end{equation}
It is well known (see, e.g. Lipster and Shiryaev (2001)\cite{lipster2001statistics} or Ceci and Colaneri (2012, 2014)\cite{ceci2012,ceci2014}) that $I_t$ is a Brownian motion with respect to the observable filtration $\bF^S$. The proposition below, proved for instance in Lipster and Shiryaev (2001)\cite{lipster2001statistics}, provides the dynamics of $\pi_t$ and $R_t$.
\begin{proposition}
The conditional mean and variance of the market price of risk satisfy the equations
\begin{equation}\label{eq:cond_exp}
d \pi_t=-\lambda (\pi_t- \bar X) d t + (R_t+\rho\sigma_X) d I_t, \quad \pi_0\in \R,
\end{equation}
\begin{align}\label{eq:R}
d R_t = \left[\sigma_X^2 -2\lambda R_t -(R_t+\rho\sigma_X)^2\right] d t, \quad R_0\in \R^+.
\end{align}
\end{proposition}
From \eqref{eq:R} we see that the conditional variance of the market price of risk is deterministic and satisfies a Riccati ordinary differential equation.
Using Equation \eqref{eq:innovation}, we also get the equivalent dynamics of the stock, the wealth and the benchmark under partial information
\begin{align}
\frac{d S_t}{S_t}&= (r + \sigma \pi_t) d t +  \sigma d I_t,\label{S_pi} \\
\frac{d W_t}{W_t}&= (r + \theta_t \sigma\pi_t) d t + \theta_t \sigma dI_t, \label{W_pi} \\
\frac{d Y_t}{Y_t}&= (r + \beta \sigma\pi_t) d t + \beta \sigma dI_t.\label{Y_pi}
\end{align}
All processes are  driven by  the innovation process  and hence the sub-market restricted to those claims that can be replicated by strategies in $\mathcal A^S$ is  complete.
We solve the optimization problem \eqref{opt_pr} using the martingale method (see, for instance Cox and Huang (1989)\cite{Cox and Huang}),  transforming the dynamic  optimization problem \eqref{opt_pr} where the control variable is a strategy into an equivalent static  problem where the control variable is the terminal wealth.

To identify the terminal wealths  reachable  from the initial budget $w$ with feasible strategies, we introduce
the unique state price density process
\begin{equation}
\frac{d \xi_t}{ \xi_t} = -r dt - \pi_t dI_t, \quad  \xi_0=1.
\end{equation}
The static optimization problem, equivalent to  \eqref{opt_pr} is
\begin{equation}
\max_{ W_T} E[u( W_T f_T ( W_T,Y_T))],
\label{static_pr_pi}
\end{equation}
with  budget constraint
\begin{equation}
w = E[\xi_T  W_T]
\label{static_BC_pi}
\end{equation}
The objective function in problem \eqref{static_pr_pi}--\eqref{static_BC_pi} is not concave in  $W_T$. To overcome this issue we  apply the concavification procedure proposed by Carpenter (2000)\cite{carpenter2000}. Following the approach in Proposition 2 of Basak et al. (2007)\cite{BPS2007}, we define the optimal final wealth relative to the benchmark, which is given by $V_T=\frac{ W^\star_T}{Y_T}$, where $W^\star_T$ is the optimal final wealth in problem \eqref{static_pr_pi}-\eqref{static_BC_pi}. One of the advantages of working with this quantity is that it has an explicit representation 
(e.g. equation $(A7)$ in Basak et al. (2007)\cite{BPS2007} or equation $(6)$ in Nicolosi et al. (2018)\cite{NHA}) given, for completeness, by equation \eqref{VT} in Appendix \ref{optfinval}. One key characteristic is that $V_T$ is a function of $\zeta_T:= \xi_T Y_T^\gamma$ only.
Computing $V_T$, enables us to characterize the optimal terminal wealth. However this is not sufficient to obtain the trading optimal portfolio strategy, for which, we need to know the value of optimal wealth at any time $t\in [0,T]$, and consequently, we must determine the relative wealth $V_t= \frac{ W^\star_t}{Y_t}$. To compute $V_t$ we consider the {\em benchmarked market}, where we discount all processes with the num\'{e}raire $Y_t$\footnote{Notice that the benchmark is a positive self-financing portfolio, and hence it can be taken as num\'{e}raire.}. Due to market completeness there exists an equivalent risk neutral probability measure $Q$ for the benchmarked market. Put in other words there exists a probability measure $Q$ which is equivalent to $P$ and such that the price process of any {\em benchmarked} traded asset (i.e. any traded asset discounted with $Y_t$), is a martingale under $Q$\footnote{Introducing the measure $Q$ allows us to circumvent technical difficulties: for instance, to get the optimal wealth $W^\star_t$ under the physical measure $P$, one should know the joint distribution of $Y^\gamma$ and $\xi$. This is unnecessary if we perform the change of measure, where one can use the martingale property and get the distribution of $W^\star_t$ more directly. See, e.g. Proposition 2 in Basak et al. (2007)\cite{BPS2007}}.

We define the process $\zeta_t = \xi_t Y_t^{\gamma}$ and derive its distribution under $Q$. Let the conditional moment generating function of $\ln(\zeta_t)$ under the measure $Q$ be given by
$$
 H(t,\zeta,\pi;z) =E^Q \left[  \zeta_T^{z} \vert  \F^S_t\right] = E^Q \left[  \zeta_T^{z} \vert \zeta_t = \zeta ,\pi_t = \pi \right]
$$
for some complex number $z$. The function $H(t,\zeta,\pi;z)$ plays a key role in solving the optimization problem (see Proposition \ref{Vtthm} below). It is characterized in the following technical lemma:
\begin{Lemma}\label{lemmaH}
Under the usual regularity conditions, the conditional moment generating function of $\ln(\zeta_T)$ under the measure $Q$  is given by:
\begin{equation}
H(t,\zeta,\pi;z) = \zeta^{z}  e^{A(t;z)+B(t;z) \pi+\frac{1}{2}
C(t;z) \pi^2}
\label{HXSTOCN}
\end{equation}
where $A(t;z)$, $B(t;z)$ and $C(t;z)$ are deterministic functions satisfying a system of Riccati Equations. 
\end{Lemma}
The proof of Lemma \ref{lemmaH} and the Riccati equations for the functions $A(t;z)$, $B(t;z)$ and $C(t;z)$ are given in Appendix \ref{lemma} and goes along the same lines as in Nicolosi et al. (2018)\cite{NHA}. We remark that in this particular case, because both the drift and volatility in the dynamics of the filter $\pi_t$  are not constant \eqref{eq:cond_exp}, the coefficients of the Riccati equations that characterize the functions $A, B$ and $C$ are time-dependent. The solutions of non-homogeneous Riccati equations are discussed in Appendix \ref{app:ODEs}.

In the next step we use Fourier Transform to compute the optimal relative wealth and the optimal strategy at any time $t\leq T$.

%

\begin{proposition}\label{Vtthm}
Let $R_j$, for  $j=1,2,3,4$ be real numbers such that $R_1  < -1/\gamma$, $ R_4 > -1/\gamma$  and
\begin{equation} \label{not_explodes}
 H(t,\zeta,\pi;R_j) = E^Q \left[  e^{ R_j \ln(\zeta_T)} \vert \zeta_t = \zeta ,\pi_t = \pi \right] <
\infty.
\end{equation}
Then,
\begin{itemize}
\item[(i)]the relative wealth $V_t$ is given by
\begin{equation}
V_t = \frac{1}{2\pi}\sum_{j=1}^4
\int_{-\infty}^{+\infty}\hat \varphi_j(u+iR_j) H(t,\zeta,\pi;R_j-iu) du\label{Vt}
\end{equation}
where 
the functions $\hat \varphi_j(z)$ of the complex variable $z$ are given in Appendix \ref{optfinval};
\item[(ii)]the optimal strategy  is
\begin{equation}
\theta_t = \beta+\frac{1}{ \sigma V_t}\left(\frac{\partial V}{\partial \zeta}\zeta_t(\gamma\beta \sigma - \pi_t)+  \frac{\partial V}{\partial \pi}  (R_t+\rho\sigma_X) \right)
\label{theta_t}.
\end{equation}
\end{itemize}
\end{proposition}

A sketch of the proof of Proposition \ref{Vtthm} is given in Appendix \ref{lemma}.
It is important to notice that Proposition \ref{Vtthm} provides a useful semi-closed solution for the optimal relative wealth and the optimal strategy.
Examples of applications of those formulas are given in the next Section \ref{NumIllustration}.


\section{A numerical illustration}
\label{NumIllustration}

In this part of the paper we  study the implications of considering parameter uncertainty  on optimal  strategies of a portfolio manager subject to implicit incentives.
The goal is to show that the impact depends on a combined effect of risk aversion and of market conditions.
In a relatively stable market (i.e. low volatility) with lower expected returns, portfolio managers tend to increase their exposure to the risky asset when underperforming the benchmark and  decrease it when overperforming. The opposite happens when the market is more volatile and expected returns are higher.
Risk-aversion has a direct influence on the view of the manager regarding the uncertain estimates. Managers with a risk-aversion parameter larger than $1$, fear that the true value may be below the current estimate and hence reduce their leverage. On the contrary, managers with risk aversion lower than $1$ hope that the true value may be higher than the current estimate and hence tilt their strategy in the opposite way.

To illustrate such behavior with an example, we consider a simplified version of our model,
where the  market price of risk is constant but unknown, and is represented by a random variable $X_0$. The manager who is uncertain about the value of the market price of risk assumes that the exact value of $X_0$ is drawn from a normal random variable with mean $\pi_0$ and variance $R_0$.
This setting is analogous  to that of Brennan (1998), where the manager does not know the  value of the drift of the price process and can only estimate its expected value $m_0$, which is related to the market price of risk $\pi_0$ by
$$\pi_0= \frac{m_0-r}{\sigma}.$$
Setting $\lambda=0$ and $\sigma_X=0$, we get   a  stochastic conditional mean
\begin{equation}
 \pi_t= R_t d I_t, \quad \pi_0 \in \R,
\end{equation}
and conditional variance
\begin{equation} \label{cond_var}
R_t=\frac{R_0}{R_0 t +1}.
\end{equation}

To highlight the effects of uncertainty of parameter estimates, as a comparison we use the strategy of a manager who believes that she knows exactly the value $X_0$ of the market price of risk. We call this manager {\em myopic} because she does not adjust her strategy to
hedge for future  changes on the estimates.
We denote by  $\theta^0_t$  the fraction of wealth invested by the myopic manager in the risky asset and by $\theta_t$ the optimal strategy of the partially informed manager given by equation \eqref{theta_t}. For comparison reason,  we also consider the Merton level  $\theta_N = \frac{1}{\gamma}\frac{\mu_0 - r}{\sigma^2} $ that is the optimal investment of the myopic manager  who  optimizes only the utility of terminal wealth, without other incentives.

Figure \ref{strategy_ec_a} represents the myopic strategy $\theta^0_t$ (dotted line) and the optimal strategy under partial information $\theta_t$ (continuous line) as functions of the relative return of the portfolio with respect to the benchmark, that is $R_t^W-R_t^Y$, at time $t=0.25$, either for $\gamma = 0.8$ (left panels) or for $\gamma = 2$ (right panels). The parameters of the implicit incentives structure at time $T=1$ in Equation \eqref{eq:compensation} are the same as in Figure \ref{PayoffFig}.

\begin{figure}[htbp]
\centering
\includegraphics[trim={2cm 0 0 0},clip, scale = 0.6]{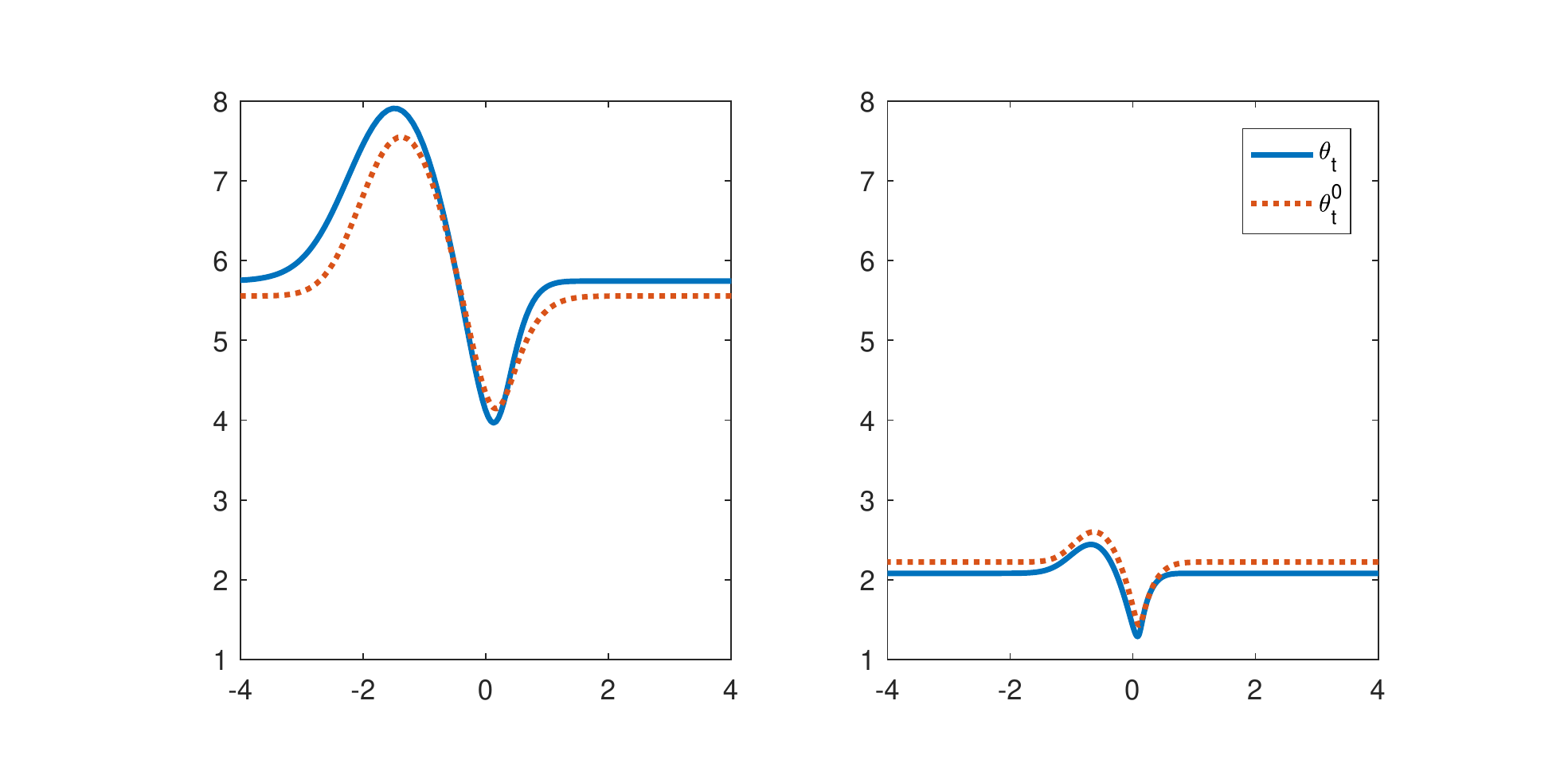}
\includegraphics[trim={2cm 0 0 0},clip, scale = 0.6]{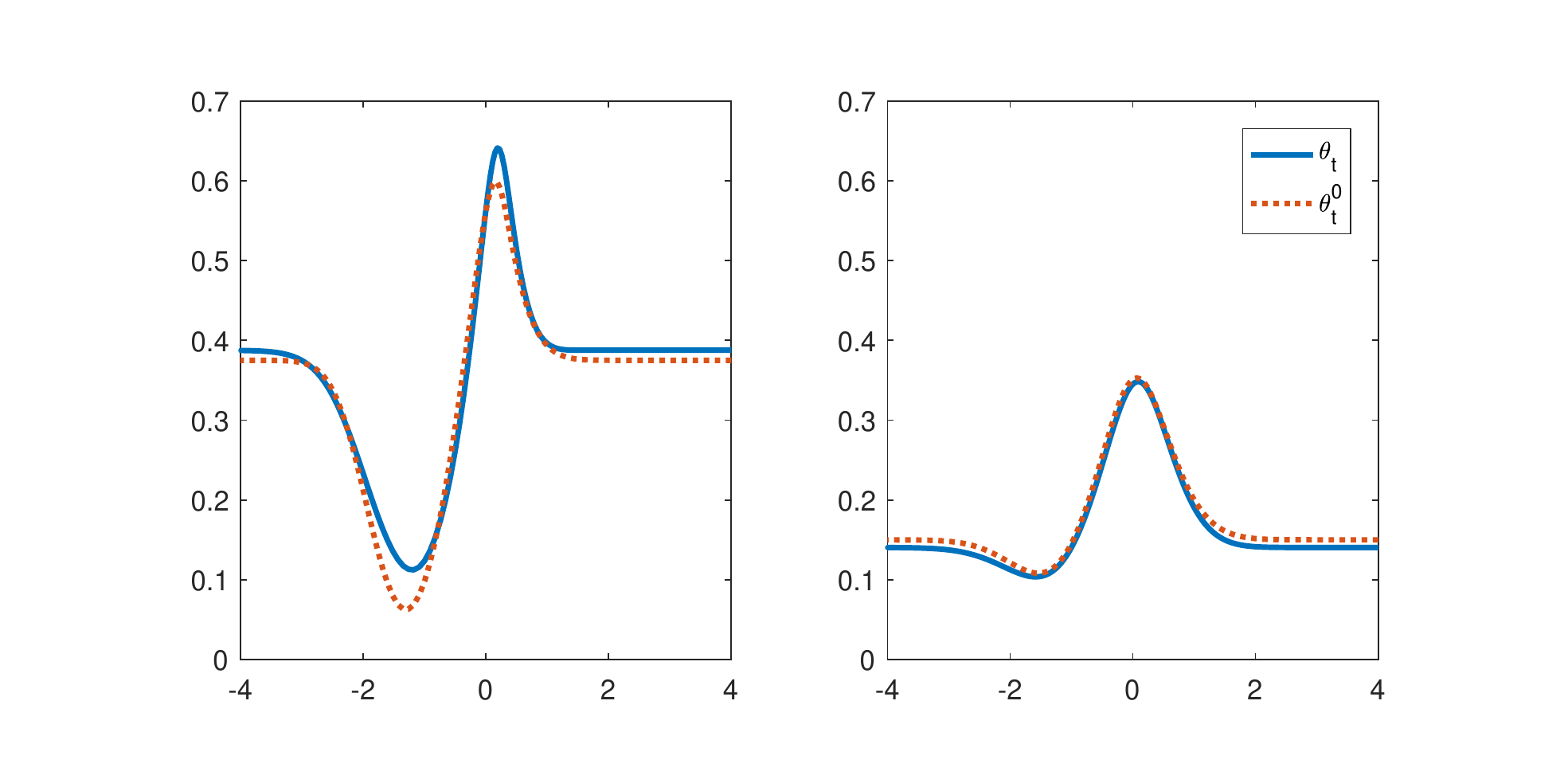}
\caption{Optimal strategies for different economies and different managers.
The optimal strategy $\theta_t$ (continuous line) and the myopic one $\theta^0_t$ (dotted line) at time $t=0.25$ are reported as functions of the relative return of the portfolio with respect to the benchmark $R_t^W-R_t^Y$.
Left panels represent the strategies of managers with risk aversion $\gamma=0.8$,
right panels those of more risk-averse managers ($\gamma = 2$).
Top panels represent economy (a), that is  a less volatile market with higher returns ($\sigma = 0.15$, $\pi_t=0.667$),  the bottom panels are referred to  economy (b), a more volatile market with lower returns ($\sigma = 1$, $\pi_t=0.3$).
The parameters of the payoff function are the same as in Figure \ref{PayoffFig}. The others parameters are $T=1$, $r = 0$, and $R_0=0.09$.}
\label{strategy_ec_a}
\end{figure}

The top panels  show the strategies when the Merton level is higher than the investment in the risky asset of the benchmark portfolio, that is when $\theta_N > \beta$. This setting corresponds to a market situation with relatively small volatility and returns, that is called Economy (a) by Basak et al. (2007),  obtained by taking $\sigma = 0.15$, $r = 0$ and $m_0 = 0.1$ in our model.

The bottom panels show the strategies when the Merton level is lower than investment in the risky asset of the benchmark, that is when $\theta_N < \beta$. This is Economy (b) in Basak et al. (2007) which describes a more volatile and remunerative market and it is obtained by setting $\sigma = 1$, $r = 0$ and $m_0 = 0.3$. The level of uncertainty on the initial estimate is given by $R_0 = 0.09$ for all the panels.

The panels on the left represent a less risk averse manager $\gamma=0.8$, those on the right a more risk-averse one $\gamma=2$.
By comparing the left to the right panels  we  see that the investment in risky asset decreases with risk aversion for both Economies (a) and (b).
The strategies for a myopic and a non-myopic manager are qualitatively similar to each other but the non-myopic manager invests always more or less than the myopic one, depending on the risk-aversion parameter. When the risk aversion parameter $\gamma$ is equal to $1$ (that is the case of  logarithmic utility), the strategies of the myopic and of the non-myopic manager coincide. A non-myopic manager with a risk aversion smaller than $1$ (left panels) tends to be more exposed to the risky asset than a myopic manager with the same risk aversion. In this case, the non-myopic manager acts optimistically, as she believes that, increasing the precision of the estimates of the market price, the correct value will be higher than the current one. Instead, the more risk-averse manager (right panels) is pessimistic and reduces the exposition to the risky asset fearing that the future estimates will be lower than the current one.
By comparing  top to  bottom panels in Figure \ref{strategy_ec_a}  we see  the effects of the overall economic condition on the optimal strategy, depending on the current results of the portfolio management strategy.
When the relative performance is either too low or too high for the incentives to have an effect on the final reward, the optimal strategy approaches a constant level that corresponds to the optimal risky exposure  without   incentives and hence the myopic investment converges  to the Merton level.  If the manager is underperforming but  still hopes to recover, or if she is slightly ahead, but still fearing to end behind, she adjusts the portfolio strategy in a way that depends on the economic conditions.
In the case of economy (a) (top panels), she  increases the exposure when trailing and  decreases it when leading. The economy (b), representing a  more volatile market and  higher expected returns (bottom panels), induces the same manager to take  opposite choices.

\section{Conclusions}\label{sec:conclusions}

We  studied a portfolio optimization problem for a manager who is compensated depending on the performance of her portfolio relative to a benchmark. The manager can invest in a risk-free asset and in a risky asset whose return depends on a latent variable representing the market price of risk. Hence she solves an optimization problem under partial information. Due to the implicit incentives given by the funds flow to relative performance, the utility function of the manager is not concave and hence existence of the optimum does not trivially hold. We solve the optimization problem using the martingale approach and a concavification procedure. This approach can be successfully applied due to completeness of the market under partial information. Optimal wealth and consequently the optimal strategy are characterized in a semi-explicit form via Fourier transform.
We  illustrated our results with an example, where we assume that the market price of risk is constant but unknown. We observed that the level of risk aversion has an influence on the manager's estimate of the market conditions, and consequently on her investment choices. Managers with a small risk aversion parameter are optimistic: they tend to increase their investment in the risky assets compared to myopic managers, believing that the true value of the market price of risk (and hence of the asset return) is more favourable than her estimate.
It is also seen that if the market is not subject to large fluctuations, managers invest more in risky assets when they are underperforming the benchmark, in anticipation to retrieve benchmark revenues, and invest less in the risky asset when overperforming, to avoid possible downward movements of the market.

\appendix
\section{Optimal final wealth}\label{optfinval}
In this section we characterize the final wealth relative to the benchmark $V_t$, for every $t \in [0,T]$. We first consider the optimal final wealth relative to the benchmark, given by the random variable $V_T= \frac{ W^*_T}{Y_T}$. Its expression has been computed in Basak et al. (2007) and also reported in Nicolosi et al. (2018), and in our framework, is given by
\begin{equation}
V_T = \varphi_1 (\zeta_T; \zeta^1) \ + \
\varphi_2 (\zeta_T; \zeta^1,\zeta^2) \ + \
\varphi_3 (\zeta_T; \zeta^2,\zeta^3) \ + \
\varphi_4 (\zeta_T; \zeta^3)
\label{VT}
\end{equation}
where $\zeta_T = \xi_T Y_T^{\gamma}$ and functions $\varphi_j$, for $j=1, \dots, 4$ are
\begin{align}
&\varphi_1 (\zeta; \zeta^1) =  f_H^{1/\gamma-1} y^{-1/\gamma} \zeta^{-1/\gamma}
\textbf{1}_{\zeta<\zeta^1}\label{varphi}\\
&\varphi_2 (\zeta; \zeta^1, \zeta^2)  =
e^{\eta_H}\textbf{1}_{\zeta^1\leq\zeta<\zeta^2}
\label{varphi2}\\
&\varphi_3 (\zeta; \zeta^2, \zeta^3)  =  h(\zeta)\textbf{1}_{\zeta^2\leq\zeta<\zeta^3}
\label{varphi3}\\
&\varphi_4 (\zeta; \zeta^3)  =  f_L^{1/\gamma-1} y^{-1/\gamma} \zeta^{-1/\gamma}
\textbf{1}_{\zeta \geq \zeta^3}.\label{varphi4}
\end{align}
The value $y\in \R$ is the Lagrange multiplier that ensures that the budget constraint of the optimization problem $ w = E_0\left[ \xi_T  W_T^*\right]$ is satisfied; the function $h(\zeta)$ is the solution of the equation
\begin{equation}
\frac{d}{d V} u(V f_L + V \psi (\ln V-\eta_L)) = y \zeta.
\label{fnh}
\end{equation}
Parameters $\zeta^1,\zeta^2$ and $\zeta^3$, and hence the value of $V_T$, depend on the following relation, called {\em Condition A}:
\begin{equation}\frac{\gamma}{1-\gamma} \left( \frac{f_H
+\psi}{f_L}\right)^{1-1/\gamma}+ \left( \frac{f_H +\psi}{f_H}\right)-
\frac{1}{1-\gamma} \geq 0 .
\label{cond1}
\end{equation}
This condition is related to concavification and defines the region in which the optimum of the maximization problem with the utility function and the optimum of the optimization problem where the utility function is replaced by the smallest concave function above it, is reached in a point where the two functions coincide.

Proposition 2 of Basak et al. (2007) shows that, if {\em Condition A} holds, then $\zeta^1,\zeta^2$ and $\zeta^3$ in  \eqref{VT} satisfy
$\zeta^1 = f_H^{1-\gamma}e^{-\gamma\eta_H}/y$, and
$\zeta^3=\zeta^2>\zeta^1$ satisfying $g(\zeta^3) = 0$,  with $$g(\zeta) =
\left(\gamma\left(\frac{y}{f_L}\zeta\right)^{1-1/\gamma}-(f_H
e^{\eta_H})^{1-\gamma}\right)/(1-\gamma)+e^{\eta_H}y\zeta.$$
Hence, in this case, $\varphi_3 (\zeta; \zeta^2, \zeta^3) $ is the indicator function of the empty set and therefore it is zero.

When {\em Condition A}  is not met, Basak et al. (2007) show in Appendix C that $\zeta^1 =  f_H^{1-\gamma}e^{-\gamma\eta_H}/y$, $\zeta^2 = (e^{\eta_H}f_H)^{-\gamma}(f_H+\psi)/y$ and
$\zeta^3 = (f_L\underline{V})^{-\gamma}f_L/y$ with  $\underline{V}$ being the left boundary of the region where the objective function is not concave. Denoting with $\overline{V}$ the right boundary of the non concave region, $\underline{V}$ and $\overline{V}$ can be computed as the points where the straight line between these two points is tangent to the objective function.

Next, we provide a representation for the function $\hat \varphi_j(z)$, $j=1,\ldots,4$, which are used to compute the optimal relative wealth $V_t$ given in  Proposition \ref{Vtthm}. The functions $\hat \varphi_j(z)$, $j=1,\ldots,4$,  are the Fourier transforms of the functions in \eqref{varphi}-\eqref{varphi2}-\eqref{varphi3}-\eqref{varphi4} and they are given by
\begin{eqnarray}
\hat \varphi_1(z) &=& f_H^{1/\gamma-1}
y^{-1/\gamma}\frac{(\zeta^1)^{-1/\gamma+iz}}{-1/\gamma+iz}
\label{phihat_1} \\
\hat \varphi_2(z) &=&  e^{\eta_H} \frac{(\zeta^2)^{iz} -
(\zeta^1)^{i z}}{iz}  \label{phihat2} \\
\hat \varphi_4(z) &=&  f_L^{1/\gamma-1}
y^{-1/\gamma}\frac{(\zeta^3)^{-1/\gamma+iz}}{-1/\gamma+iz}.\label{phihat3}
\nonumber
\end{eqnarray}
Numerical computations of the Fourier transform $\hat \varphi_3 (z)$, which are needed only when {\em Condition A} is not satisfied, are given in Section 4.1 of Nicolosi et al. (2018)\cite{NHA}.

\section{Proofs} \label{lemma}

This section contains the proofs of Lemma \ref{lemmaH} and Proposition \ref{Vtthm}.

\begin{proof}[Proof of Lemma  \ref{lemmaH}]
We define the process
$$ I^Q_t=I_t- \int_0^t (\sigma \beta - \pi_s)ds.$$
By Girsanov Theorem this is a $Q$-brownian motion (see, e.g. Chap. 26 of Bj\"ork (2009)\cite{B}).
Then the $Q$-dynamics of the filter $\pi_t$ is given by
\begin{equation*}
d\pi_t = \left( \lambda (\bar{X} - \pi_t ) +(R_t+\rho \sigma_X) (\beta \sigma
-\pi_t)\right)dt +  (R_t+\rho\sigma_X) dI^Q_t.
\end{equation*}
Using Ito's Lemma, we get that $\zeta_t= \xi_t Y^\gamma$ under $Q$ has the following dynamics
\begin{eqnarray*}
\frac{d\zeta_t}{\zeta_t} &=& \left(r(\gamma-1) +
\frac{1}{2}\gamma(\gamma+1) \beta ^2 \sigma^2 -
(\gamma+1)\beta \sigma \pi_t+\pi_t^2 \right) dt \\
&+&(\gamma\beta \sigma -\pi_t)dI_t^Q.
\end{eqnarray*}
Since the process $H(t,\zeta_t,\pi_t;z)$  is a $(\bF^S, Q)$-martingale, equating the $dt$-term to zero leads to the partial differential equation (for simplicity we drop the
arguments of the functions)
\begin{eqnarray}
0 &=& \frac{\partial H}{\partial t} +\frac{\partial H}{\partial
\zeta}\zeta \left( r(\gamma-1) +
\frac{1}{2}\gamma(\gamma+1)\beta^2 \sigma^2  -
(\gamma+1)\beta \sigma \pi+\pi ^2  \right) \nonumber\\
&+&\frac{\partial H}{\partial \pi} \left( \lambda_X (\bar{X} - \pi ) +
(R_t+\rho\sigma_X) (\beta \sigma -\pi) \right)+
\frac{1}{2} \frac{\partial^2 H}{\partial
\pi^2} (R_t+\rho\sigma_X)^2\nonumber \\
&+&\frac{\partial^2H}{\partial \zeta\partial \pi}\zeta (R_t+\rho\sigma_X) (\gamma \beta
\sigma -\pi) + \frac{1}{2}\frac{\partial^2 H}{\partial \zeta^2}\zeta^2
(\gamma\beta \sigma -\pi)^2
\label{PDEXSTOCN}
\end{eqnarray}
with the boundary condition at
time $T$
\begin{equation}  \label{boundary_condition}
H(T,\zeta,\pi;z) = \zeta^{z}, \quad \zeta \in \mathbb{R}^+, \pi \in \mathbb{R}, z \in
\mathbb{C}.
\end{equation}
We use a similar approach as in the optimization problem under full information (see Nicolosi et al. (2018)\cite{NHA}), and consider an exponential-polynomial ansatz of the type
\begin{equation}
H(t,\zeta,\pi;z) = \zeta^{z}  e^{A(t;z) + B(t;z) \pi+\frac{1}{2}
C(t;z) \pi^2}
\label{HXSTOCN2}
\end{equation}
where $A(t;z)$, $B(t;z)$ and $C(t;z)$ are deterministic functions. From the boundary condition \eqref{boundary_condition} we
get that
\begin{equation}
A(T;z)= 0, \quad B(T;z)=0, \quad C(T;z) = 0.
\label{BC_X}
\end{equation}
Moreover, substituting the partial derivatives of the function $H$ 
into \eqref{PDEXSTOCN} and imposing that coefficients of $\pi^2$, $\pi$ and the
constant terms are equal to zero, we obtain the  system of ordinary differential equations for $A(t;z)$, $B(t;z)$ and $C(t;z)$
\begin{align}
\frac{\partial C}{\partial t} =& -(R_t+\rho\sigma_X)^2 C^2 + 2 (\lambda_X +
(1+z) (R_t+\rho\sigma_X)) C-z(z+1)
\label{R1} \\
\frac{\partial B}{\partial t} =& \left(\lambda_X + (z+1) (R_t+\rho\sigma_X) -
(R_t+\rho\sigma_X)^2 C   \right)B \nonumber\\
&-  \left(\lambda_X\bar{X} + (1+z \gamma) \beta \sigma (R_t+\rho\sigma_X)
\right) C +z (1+z\gamma) \beta \sigma  \label{R2}  \\
\frac{\partial A}{\partial t} =&  z r(1-\gamma) - \frac{1}{2}z \gamma
(1+z \gamma) \beta^2 \sigma^2 -\frac{1}{2} (R_t+\rho\sigma_X)^2 B^2  \nonumber  \\
&- ( \lambda_X \bar{X} + (1+z \gamma) \beta \sigma (R_t+\rho\sigma_X)) B - \frac{1}{2} (R_t+\rho\sigma_X)^2 C. \label{R3}
\end{align}
Notice that this is a coupled system of equations of Riccati type, with non-homogeneous coefficients.
\end{proof}

\begin{proof}[Proof of Proposition \ref{Vtthm}]
The proof of part $(i)$ follows the same lines of  Nicolosi et al. (2018)\cite[Proposition 2.1]{NHA}. Here we summarize the idea. Since the market model under partial information is complete, after applying concavification we get that the relative final wealth $V_T$ is given by the formula \eqref{VT} in Appendix \ref{optfinval}.
Plugging the expression of $V_T$ into $V_t=E^Q\left[V_T|\F^S_t\right]$ and then using Fourier transform we can calculate the value at  time $t$ of the optimal relative value.
For part $(ii)$, we first determine the dynamics of $ W^\star_t = Y_t V_t= Y_t V(t,\zeta_t,\pi_t)$ via Ito's product rule. Then comparing this equation with equation \eqref{W_pi} provides the expression for $\theta_t$ in \eqref{theta_t}.
Notice that the integrals in \eqref{Vt} are principal value integrals and the partial derivatives of the function $V$ in \eqref{theta_t} can be  computed from \eqref{Vt} by taking the derivative under the integral sign.
\end{proof}

\section{Solutions to non-Homogeneous Riccati ODEs}\label{app:ODEs}

We discuss the solution of the system of non-homogeneous system of Riccati equations arising in the expression of the conditional moment generating function of $\ln(\zeta_T)$.
Precisely, we show how to solve the system of equations \eqref{R1} -- \eqref{R2}. Equation \eqref{R3} can be computed by direct integration, and we do this numerically. Following, for instance, Brendle (2006) and Colaneri et al. (2020), it can be proved that the functions $B$ and $C$ satisfy

\begin{eqnarray}
C(t;z) &=& \frac{C^o(t;z)}{1+\frac{1}{z} C^o(t;z) R_t} \nonumber \\
B(t;z) &=& \frac{B^o(t;z)}{1+\frac{1}{z} C^o(t;z) R_t} \nonumber
\end{eqnarray}
for some functions $C^o(t;z)$ and $B^o(t;z)$ which solve the homogeneous system of Riccati equations below
\begin{eqnarray}
\frac{\partial C^o}{\partial t} &=& \left(\frac{1}{z}(1-\rho^2)-\rho^2 \right) \sigma_X^2 C^{o2} + 2 (\lambda_X +
(1+z) \rho\sigma_X) C^o-z(z+1)
\label{Rom1} \\
\frac{\partial B^o}{\partial t} &=& \left(\lambda_X + (z+1) \rho\sigma_X +
 \left(\frac{1}{z}(1-\rho^2)-\rho^2 \right) \sigma_X^2 C^o   \right) B^o \nonumber\\
&-&  \left(\lambda_X \bar{X} + (1+z \gamma) \beta \sigma \rho\sigma_X
\right) C^o +z (1+z\gamma) \beta \sigma  \label{Rom2}
\end{eqnarray}
with boundary conditions
\begin{equation}
B^o(T,z)=0, \quad C^o(T,z) = 0.
\label{BC_F}
\end{equation}
Equations \eqref{Rom1}--\eqref{Rom2} have a solution in closed form\footnote{Equations \eqref{Rom1}--\eqref{Rom2} are related to the conditional moment generating function of the process $\ln(\zeta_T)$, under full information. In fact, by Markovianity it holds that
$$
 \widetilde{H}(t,\zeta,\pi;z) :=E^Q \left[\zeta_T^{z} \vert  \F_t\right].
$$
Here using the ansatz
$\widetilde{H}(t,\zeta,\pi;z)= \zeta^{z}  e^{A^o(t;z) + B^o(t;z) \pi+\frac{1}{2}
C^o(t;z) \pi^2}$
we get that $B^o(t;z)$ and $C^o(t;z)$ solve \eqref{Rom1}--\eqref{Rom2} with the boundary condition \eqref{BC_F} and  $A^o(t;z)$ satisfies
\begin{align*}
\frac{\partial A^o}{\partial t} &=  z r(1-\gamma) - \frac{1}{2}z \gamma
(1+z \gamma) \beta^2 \sigma^2 -  ( \lambda_X \bar{X} + B^o (1+z \gamma) \beta \sigma \sigma_X)  - \frac{1}{2} \sigma_X^2 (C^o +B^{o2}),
\end{align*}
with the boundary condition $A^o(T;z)=0$.} see for instance Filipovi\'{c} (2009)\cite[Lemma 10.12]{filipovic2009}.

{}


\begin{thebibliography}{}

\bibitem{altay2018} Altay, S., Colaneri, K., and  Eksi, Z. (2018). Pairs trading under drift uncertainty and risk penalization. International Journal of Theoretical and Applied Finance, 21(7).

\bibitem{altay2019} Altay, S., Colaneri, K., and Eksi, Z. (2020). Optimal Converge Trading with Unobservable Pricing Errors. The Annals of Operations Research. https://doi.org/10.1007/s10479-020-03647-z.

\bibitem{BarucciMarazzina2015}
Barucci, E., and Marazzina, D. (2015). Risk seeking, nonconvex remuneration and regime switching. {\it International Journal of Theoretical and Applied Finance} 18(02),
DOI: http://dx.doi.org/10.1142/S0219024915500090.

\bibitem{BM2016} Barucci, E., and Marazzina, D. (2016). Asset management, High Water Mark and flow of funds. Operations Research Letters, 44(5), 607--611.

\bibitem{BMM2019} Barucci, E., Marazzina, D., and  Mastrogiacomo, E. (2019). Optimal investment strategies with a minimum performance constraint. Annals of Operations Research, 1--25.

\bibitem{BP2013} Basak, S., and Pavlova, A. (2013). Asset Prices and Institutional Investors. American Economic Review. 103 (5): 1728--1758

\bibitem{BPS2007}
Basak, S., Pavlova, A., and Shapiro, A. (2007).  Optimal asset allocation
and risk shifting in money management. {\it The Review of Financial Studies}
20(5), 1583-1621.

\bibitem{BPS2008}Basak, S., Pavlova, A., and Shapiro, A. (2008). Offsetting the implicit incentives:Benefits of benchmarking in money management. Journal of Banking and Finance. 32: 1883--1893





\bibitem{B} Bj\"ork, T. (2009). Arbitrage theory in continuous

\bibitem{BDL}
Bj\"ork, T., Davis, M. H. and Landen, C. (2010). Optimal investment under partial information. Mathematical Methods of Operations Research, 71(2), 371--399.

\bibitem{brendle2006} Brendle S. (2006). Portfolio selection under incomplete information. Stochastic Processes and their Applications,  116 (5), 701--723.

\bibitem{brennan1998} Brennan M. (1998). The role of learning in dynamic portfolio decisions. Eurepean Financial Review, 1, 295--306.


\bibitem{buraschi2014} Buraschi, A., Kosowski, R., and Sritrakul, W. (2014). Incentives and endogenous risk taking: A structural view on hedge fund alphas. The Journal of Finance, 69(6), 2819-2870.


\bibitem{carpenter2000}Carpenter, J. N. (2000). Does option compensation increase managerial risk appetite?. The journal of finance, 55(5), 2311-2331.

\bibitem{ceci2012} Ceci, C. and Colaneri, K. (2012). Nonlinear filtering for jump diffusion observations. Adv. Appl. Probab., 44(3),
678-701.

\bibitem{ceci2014} Ceci, C. and Colaneri, K. (2014). The Zakai equation of nonlinear filtering for jump-diffusion observations: existence and uniqueness. Applied Mathematics \& Optimization, 69(1), 47-82.

\bibitem{CP2009} Chen, H. L., and Pennacchi, G. G. (2009). Does prior performance affect a mutual fund's choice of risk? Theory and further empirical evidence. Journal of Financial and Quantitative Analysis, 745-775.

\bibitem{chevalier1997} Chevalier, J.,  Ellison, G. (1997). Risk taking by mutual funds as a response to incentives. Journal of political economy, 105(6), 1167-1200.



\bibitem{CHN2020} Colaneri, K., Herzel, S. and Nicolosi, M.  (2020). The value of knowing the market price of risk. Ann Oper Res. https://doi.org/10.1007/s10479-020-03596-7


\bibitem{Cox and Huang}
Cox, J.C. and Huang, C.F. (1989). Optimal Consumptions and Portfolio Policies when Asset Prices Follow a Diffusion Process. {\it Journal of Economic Theory}, 49, 33-83.


\bibitem{CK2011} Cuoco, D., and Kaniel, R. (2011). Equilibrium prices in the presence of delegated portfolio management. Journal of Financial Economics, 101(2), 264-296.






\bibitem{filipovic2009} Filipovic, D. (2009). Term-Structure Models. A Graduate Course. Springer.


\bibitem{fleming1982optimal} Fleming, W. H. and Pardoux, {\'E}. (1982). Optimal control for partially observed diffusions, SIAM Journal on Control and Optimization,
 20 (2), 261--285.









\bibitem{hata2018} Hata, H. and  Sheu, S. (2018). An optimal consumption and investment problem with partial information, Advances in Applied probability, 50,  131-153.

\bibitem{nic2018} Nicolosi, M. (2018). Optimal strategy for a fund manager with option compensation. Decisions Econ Finan 41, 1–17. https://doi.org/10.1007/s10203-017-0204-x

\bibitem{HN}
Herzel, S., Nicolosi, M. (2019). Optimal strategies with option compensation under mean reverting returns or volatilities. Comput Manag Sci 16, 47–69. https://doi.org/10.1007/s10287-017-0296-3



\bibitem{HKH2012} Huang, J. C., Wei, K. D., and Yan, H. (2012). Investor learning and mutual fund flows. In AFA 2012 Chicago Meetings Paper.





\bibitem{lackner1995} Lakner, P. (1995). Utility maximization with partial information. Stochastic processes and their applications, 56 (2), 247--273.

\bibitem{lackner1998} Lackner, P. (1998). Optimal trading strategy for an investor: the case of partial information. Stochastic Processes and their Applications, 76 (1), 77--97.

\bibitem{LWY2013} Lan, Y., Wang, N., and Yang, J. (2013). The economics of hedge funds. Journal of Financial Economics, 110(2), 300-323.

\bibitem{lee2016} Lee, S., and Papanicolaou A. (2016). Pairs trading of two assets with uncertainty in co-integration's level of mean reversion. International Journal of Theoretical and Applied Finance 19 (8), 1650054.


\bibitem{lind2016} Lindensjo, K. (2016). Optimal investment and consumption under partial information. Mathematical Methods of Operations Research, 83, 87--107.


\bibitem{lipster2001statistics}
Lipster, R.S. and Shiryaev, A.N. (2001). Statistics of Random Processes. Springer Verlag, Berlin.

\bibitem{MTG2019} Ma, L., Tang, Y. and Gomez, J. (2019). Portfolio manager compensation in the US mutual fund industry. The Journal of Finance. 76 (2): 587--638


\bibitem{merton1971} Merton, R. C. (1971). Optimum Consumption and Portfolio Rules in a Continuous-Time Model. Journal of Economic Theory 3:373–413.

\bibitem{NHA}
Nicolosi, M., Angelini, F. and Herzel, S. Ann Oper Res (2018) 266: 373. https://doi.org/10.1007/s10479-017-2535-y





\bibitem{xia2001} Xia, Y. (2001). Learning about predictability: the effects of parameter uncertainty on dynamic asset allocation.
Journal of Finance, 56, 205--246.


\end{thebibliography}
 \end{document}